\documentclass{llncs}
\usepackage[dvips,dvipdfm]{graphics}
\usepackage{url}
\usepackage{ascmac,bm}
\usepackage{amsmath,amssymb}
\usepackage{float}
\usepackage{hyperref}
\usepackage{algorithm,algorithmic}

\input{macros}

\begin{document}
\title{Text Compression using Abstract Numeration System on a Regular
Language}
\author{Ryoma Sin'ya\inst{1}}
\institute{Department of Mathematical and Computing Sciences, Tokyo Institute of Technology\\\email{shinya.r.aa@m.titech.ac.jp}}

\maketitle
\begin{abstract}
An abstract numeration system (ANS) is a numeration system that provides
a one-to-one correspondence between the natural numbers and a regular language.
In this paper, we define an ANS-based compression as an extension of this
correspondence. In addition, we show the following results: 1) an average
compression ratio is computable from a language, 2) an ANS-based compression
runs in sublinear time with respect to the length of the input string, and 3)
an ANS-based compression can be extended to block-based compression
using a factorial language.
\end{abstract}

\section{Introduction}\label{sec:introduction}
Compression has been closely studied in the past and remains an important
issue in the real world today.
String compression methods that focus on the input
distribution, such as entropy-based compression, have been thoroughly investigated.
Grammar-based compressions have also been investigated.
This latter method has been formalized as the {\it smallest grammar problem}, which
constructs a minimal context-free grammar that can generate a
given string \cite{DBLP:journals/tit/CharikarLLPPSS05}.

Whereas these are the two major approaches to compression, an alternative
approach uses the ``ordering'' of a given string in the given
language, namely its {\it ranking}.
Goldberg and Sipser define ranking-based compression and show that
ranking in a context-free grammar is computable in
polynomial time \cite{Goldberg:1985:CR:22145.22194}. 
A summary of the ranking-based compression approach is available
in the literature \cite{Li:2008:IKC:1478784}.
In another approach to ranking, a
one-to-one correspondence between the natural numbers and a regular language,
namely an {\it abstract numeration system} (ANS), has been proposed independently
by Lecomte and Rigo \cite{DBLP:journals/corr/cs-OH-9903005}.
The mathematical properties and abilities of an ANS were studied in
\cite{Berth:2010:CAN:1941063}. 

The present paper presents an ANS-based compression on a regular language
that uses base conversion between two regular languages, and
lays the foundations of this approach. There are three main results from our research.
First, we analyze the complexity of the algorithm for an ANS (see Table \ref{table:complexity}, \S \ref{sec:algorithm}).
Second, we can find the average compression ratio of an ANS-based compression
(see Theorem \ref{theorem:ratio}, \S \ref{sec:compression_ratio}).
Third, we propose a block-based compression method that does not reduce the
compression ratio by using a factorial language (see Theorem \ref{theorem:sepcomp}, \S \ref{sec:sepcomp}).
\section{Preliminaries}\label{sec:preliminaries}
We first introduce some preliminary notions, notation, and theorems
from the literature \cite{Sakarovitch:2009:EAT:1629683,DBLP:journals/corr/abs-1010-5456,Berth:2010:CAN:1941063}.

\subsection{Regular languages and automata}
The set of nonnegative integers is denoted by $\nat \ni 0$, and the set of
positive real numbers is denoted by $\real$.
Let $A$ denote a finite set of symbols, with $A^n$ denoting the set
of all strings of length $n$ over $A$.
Let $A^*$ represent the set of all strings over $A$:
$A^* := \bigcup_{i=0}^{\infty} A^i.$
The string of length 0 is $\varepsilon$ ($A^0 = \{\varepsilon\}$).
A subset $L$ of $A^*$ is called a language over $A$.
The characters $w$ and $\sigma$ are used to denote a string $w \in A^*$ and
a symbol $\sigma \in A$. In addition, $|w|$ denotes the length of $w$ and
$w_i$, such that $1 \leq i \leq |w|$ denotes the $i$-th symbol of $w$.
The power set of a set $S$ is written $\powerset(S)$.
A finite automaton $\CA$ is a 5-tuple $\CA = (Q, A, \delta, I, F)$, comprising
a finite set of states $Q$, a finite set of input symbols $A$,
a transition function $\delta: Q \times A \rightarrow \powerset(Q)$, and a
set of start and accepting states $I \subseteq Q$ and $F \subseteq Q$.
The symbol like $|\CA|$ denotes the number of states for the finite automaton $\CA$. The set
of all acceptable strings of $\CA$ is denoted by $L(\CA)$.

\begin{definition}[\cite{Sakarovitch:2009:EAT:1629683}]\label{def:ufa}
An automaton $\CA = (Q, A, \delta, I, F)$ is an {\it unambiguous} finite
 automaton (UFA) when it satisfies the following two conditions:
\begin{enumerate}
 \item for every pair of states $(p, q)$ and every string $w$ in $A^*$,
	   there exists at most one transition from $p$ to $q$ on input $w$
	   ($q \in \delta(p, w)$);
 \item for every string $w$ in $L(\CA)$, there exists a unique $p$ in
	   $I$ and a unique $q$ in $F$ such that $q \in \delta(p, w)$. \ed
\end{enumerate}
\end{definition}

Any {\it deterministic} finite automaton (DFA) will clearly satisfy the UFA conditions. 
To emphasize that an automaton is deterministic, we use $\CD$ to denote a
DFA.
Similarly, we use $\CU$ to denote a UFA, with $\CN$ being used to denote a {\it
nondeterministic} finite automaton (NFA).
We will now focus on UFAs and DFAs. Although these classes are more limited than the NFA class, they are easy
to handle, and are compatible with counting functions because of their unambiguity. 

We introduce a {\it matrix representation}
for an automaton.
Because an automaton can be viewed as a digraph, its adjacency matrix
can be derived naturally.

\begin{definition}[\cite{Sakarovitch:2009:EAT:1629683}]\label{def:matrix_representation}
The adjacency matrix $\bm{M}(\CA) \in \nat^{|\CA| \times |\CA|}$ of an
 automaton $\CA$ is defined as $\bm{M}(\CA)_{pq} := \#\left\{ \sigma \in
 A \mid q \in \delta(p, \sigma) \right\}$.
The initial (row) vector $\bm{V}_I(\CA) \in \bool^{|\CA|}$ and accepting (column)
vector $\bm{V}_F(\CA)^T \in \bool^{|\CA|}$ are defined as follows:
\begin{eqnarray*}
\bm{V}_I(\CA)_q := \text{if} \; q \in I \;\; \text{then} \; 1 \;
 \text{else} \; 0, \;\;\;\;\;\; \bm{V}_F(\CA)_q &:=& \text{if} \; q \in F \; \text{then} \; 1 \; \text{else} \; 0.
\end{eqnarray*}
A triple $(\bm{M}(\CA), \bm{V}_I(\CA), \bm{V}_F(\CA))$ is called a
 matrix representation of $\CA$.
\ed
\end{definition}

In this paper, we will denote the set of eigenvalues of a matrix $\bm{M}$ by
$\lambda(\bm{M})$, and denote the zero vector by $\bm{0}$.

\subsection{Combinatorial complexity}\label{sec:combinatorial}
In formal language theory, {\it combinatorial complexity} is formulated
in terms of the complexity of the {\it counting function}\footnote{Also called
the {\it complexity function} or the {\it growth function}.} $C_L$ of a language
$L$ \cite{DBLP:journals/corr/abs-1010-5456,Berth:2010:CAN:1941063,journals/mst/ChoffrutG95a}, which is defined as follows:\;\;
$C_L(n) := \# \left\{ w \in L \mid |w| = n \right\} = \#(L \cap A^n),\; C_L^\leq(n) := \sum_{i=0}^n C_L(i)$.
We note that a counting function $C_L$ can be calculated by using a trim UFA.

\begin{lemma}[\cite{Berth:2010:CAN:1941063}]
\label{lemma:counting}
Consider a UFA $\CU$ such that $L(\CU) = L$. Then the following equation holds:
\[
 C_L(n) = \bm{V}_I \bm{M}^n \bm{V}_F.
\]
\end{lemma}
\begin{proof}
The value $\bm{M}^n_{pq}$ of the adjacency matrix of the UFA $\CU$ is
exactly the number of paths of length $n$ from state $p$ to state $q$.
Each value in vector $\bm{V}_I \bm{M}^n$ is therefore the number of paths
of length $n$ from an initial state to one of the states.  
Consequently, the value obtained by multiplying the accepting vector by the
vector $\bm{V}_I \bm{M}^n$, $\bm{V}_I \bm{M}^n \bm{V}_F$ is the number
 of acceptable paths of length $n$.
Because $\CU$ is unambiguous, there exists a one-to-one correspondence
between paths and strings. \qed
\end{proof}

Next, we introduce the important theorem about the asymptotic growth of
a counting function for a regular language \cite{Shur:2008:CCR:1813695.1813728}.
(This theorem will be needed in \S \ref{sec:compression_ratio}.)
Here, we use the standard $\order$-, $\Omega$-,
$\Theta$-notations, and the following more complicated definitions.

\begin{definition}[\cite{Shur:2008:CCR:1813695.1813728}]\label{def:shurno}
The notation $f(n) = \bar\Theta(g(n))$ means that $f(n) = \order(g(n))$ and
 there exists a sequence $\{n_i\}_1^\infty$ and a constant $c > 0$ such
 that $f(n_i) \geq c \cdot g(n_i)$ for all $i$.
That is, the notation $\bar\Theta(g(n))$ is weaker than the notation
 $\Theta(g(n))$.
\ed
\end{definition}

\begin{definition}[\cite{DBLP:journals/corr/abs-1010-5456}]\label{def:index}
Consider an automaton $\CA$. The {\it index} of $\CA$: $\ind(\CA)$ is defined as
 the maximum eigenvalue\footnote{{\it Frobenius root}
 \cite{Shur:2008:CCR:1813695.1813728}.} of the adjacency matrix of
 $\CA$:
\[
 \ind(\CA) := \max \{ \alpha \in \lambda(\bm{M}(\CA)) \}.
\] 
The {\it polynomial index} of $\CA$: $\pd{\CA}$ is defined
 as the value obtained by subtracting 1 from the maximum number of
 strongly connected components (SCCs) of index $\ind(\CA)$ connected
 by a simple path in $\CA$. \ed
\end{definition}

\begin{theorem}[Shur \cite{Shur:2008:CCR:1813695.1813728}]
\label{theorem:shur}
Let a regular language $L$ be recognized by a trim DFA $\CD$. Then,
$C_L(n) = \bar\Theta(n^\pd{\CD} \ind(\CD)^n)$ holds.
\end{theorem}

\begin{remark}
\label{remark:shur}
We omit the proof here, but note some known algorithmic
results for graphs and matrices.
Let $G$ be a digraph with $n$ vertices and $e$ paths. Then, 
SCC decomposition is computable in $\order(n + e) = \order(n^2)$.
The index of $G$ is approximately computable in $\order(n^3)$ using the classical {\it power method} \cite{Golub:1996:MC:248979}.
\ed
\end{remark}

Because the definitions of $\ind(\CD)$ and $\pd{\CD}$ are slightly complicated, it might be
hard to grasp the meaning of Theorem \ref{theorem:shur}.
We therefore introduce the following simple examples of Lemma
\ref{lemma:counting} and Theorem \ref{theorem:shur}.

\begin{figure}[t]
\begin{minipage}[t]{0.5\columnwidth}
\centering\scalebox{.6}{\includegraphics{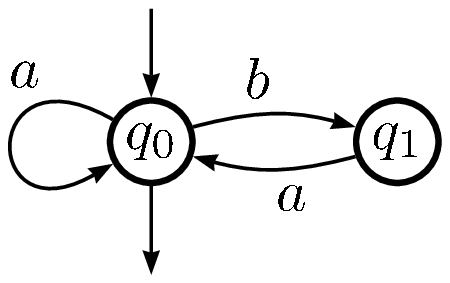}}
\caption{$\CD_{fib}$ such that $L(\CD_{fib}) = \{ a, ba \}^*$}
\label{fig:ex1}
\end{minipage}
\begin{minipage}[t]{0.5\columnwidth}
\centering\scalebox{.6}{\includegraphics{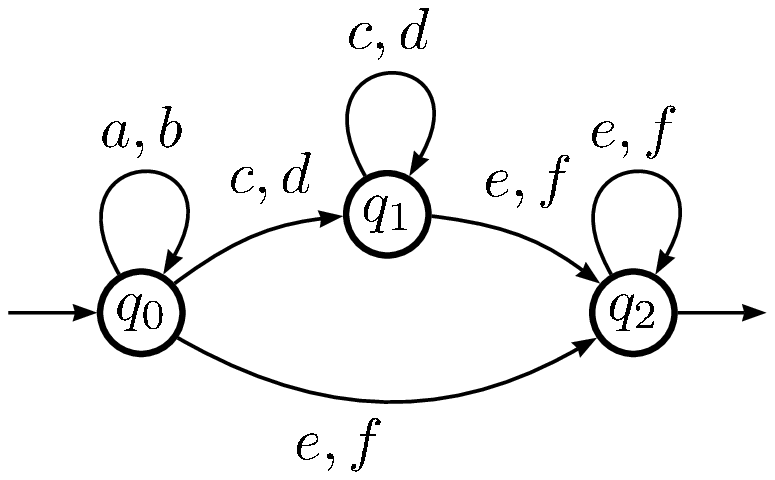}}
\caption{$\CD_{poly}$ such that $L(\CD_{poly}) = \{a,b\}^*\{c,d\}^*\{e,f\}^+$}
\label{fig:ex3}
\end{minipage}
\end{figure}

\begin{example}\label{ex1}
For the DFA $\CD_{fib}; \; L(\CD_{fib}) = \{a, ba\}^*$, as shown in Figure~\ref{fig:ex1},
consider the counting function $C_L$.
The matrix representation of $\CD_{fib}$ is:
\begin{eqnarray*}
\bm{M} = \begin{pmatrix}
				1 & 1\\
				1 & 0
				\end{pmatrix}, \;\;\;\; \bm{V}_I = \bm{V}_F^T = (1, 0).
\end{eqnarray*}
We can enumerate the sequence $C_L(n)$ up to five, as follows:
\begin{eqnarray*}
C_L(0) &=& \#\left\{ \varepsilon \right\} = 1, \;\;\;\; C_L(1) =
 \#\left\{ a \right\} = 1, \;\;\;\; C_L(2) = \#\left\{ aa, ba \right\} = 2,\\
C_L(3) &=& \#\left\{ aaa, aba, baa \right\} = 3, \;\;\;\; C_L(4) = \#\{
 aaaa, aaba, abaa, baaa, baba \} = 5.
\end{eqnarray*}
The adjacency matrix $\bm{M}$ is diagonalizable via a matrix $P$ because
$\bm{M}$ has the set of eigenvalues $\lambda(\bm{M})
= \left\{\alpha = \frac{1+\sqrt{5}}{2},\beta
=\frac{1-\sqrt{5}}{2} \right\}$. Here, we obtain the following equation
via Lemma \ref{lemma:counting}:
\begin{eqnarray*}
C_L(n) &=& \bm{V}_I \bm{M}^n \bm{V}_F = \bm{V}_I \bm{PD}^n\bm{P}^{-1}
 \bm{V}_F = \frac{1}{\sqrt{5}} \bm{V}_I \begin{pmatrix}
						\alpha^{n+1} - \beta^{n+1} & \alpha^n -
						\beta^n\\[2pt]
						\alpha^n - \beta^n & \alpha^{n-1} - \beta^{n-1}
					   \end{pmatrix} \bm{V}_F\\
&=& \frac{1}{\sqrt{5}}\left\{ \left( \frac{1+\sqrt{5}}{2}
								  \right)^{n+1} - \left(
								  \frac{1-\sqrt{5}}{2}
 \right)^{n+1}\right\}.
\end{eqnarray*}
From this, we find that $C_L(n)$ is the $(n+1)$-th Fibonacci
 number. \ed
\end{example}

\begin{example}
\label{ex2}
The index of $\CD_{fib}$ is $\frac{1+\sqrt{5}}{2}$. In addition,
the polynomial index of $\CD_{fib}$ is 0 because $\CD_{fib}$ contains
only one SCC, as does $\{q_0\}$.  From Theorem \ref{theorem:shur}, the
 following equation therefore holds. This matches the result of Example~\ref{ex1}. 
\[
 C_L(n) = \bar\Theta(n^\pd{\CD_{fib}} \ind(\CD_{fib})^n) = \bar\Theta\left(\left(\frac{1+\sqrt{5}}{2}\right)^n\right).
\]
\ued
\end{example}

\begin{example}\label{ex3}
For the DFA $\CD_{poly}$ such that $L(\CD_{poly}) = \{a,b\}^*\{c,d\}^*\{e,f\}^+$,
 as shown in Figure \ref{fig:ex3}, each state is a SCC
 and there are two self-loop transitions respectively. The index of
 $\CD_{poly}$ is therefore $\ind(\CD_{poly}) = 2$.
 Furthermore, the set of the SCCs that are connected by a simple path, except for each self-loop, are: 
 $\{\{q_0\},\{q_1\}\}, \;\; \{\{q_1\},\{q_2\}\}, \;\;
 \{\{q_0\},\{q_2\}\}, \;\; \{\{q_0\}, \{q_1\}, \{q_2\}\}.$
 The polynomial index of the DFA $\CD_{poly}$ is therefore $\pd{\CD_{poly}} =
 \#\{\{q_0\}, \{q_1\}, \{q_2\}\} - 1 = 2$. Consequently, $C_L(n)
 = \bar\Theta(n^2 2^n)$ holds from Theorem \ref{theorem:shur}.\ed
\end{example}

\begin{fact}[\cite{Shur:2008:CCR:1813695.1813728}]\label{fact:index}
For any DFA $\CD$ such that $L(\CD) = L$, its index $\ind(\CD)$ is 0 or
 greater than or equal to 1
 because the elements of adjacency matrices are positive integers.
 Moreover, there are only three possible patterns:
 1) $\ind(\CD) = 0 \Leftrightarrow$ $L$ is a finite set,
 2) $\ind(\CD) = 1 \Leftrightarrow$ $C_L(n)$ has polynomial growth, or
 3) $\ind(\CD) > 1 \Leftrightarrow$ $C_L(n)$ has exponential growth.
 In addition, the index of $\CD$ equals the maximum index of each SCC
 in $\CD$.\ed
\end{fact}

Theorem \ref{theorem:shur} clarifies the asymptotic growth of the counting
function. We now know that the index and polynomial index of a DFA $\CD$ are
not simply values, but are specific values of a regular language $L$
recognized by the DFA $\CD$. 
For this reason, we introduce the following additional definition.
\setcounter{definition}{3}
 \begin{definition}[addition]\label{redef:index}
  Let a regular language $L$ be recognized by a DFA $\CD$,
\[
\ind(L) := \ind(\CD), \;\;\;\; \pd{L} := \ind(\CD).
\]
\ued
\end{definition}

\subsection{Abstract numeration system}
A {\it numeration system} is a system that represents a number as a string
and vice versa, and is an area of mathematics in itself.
Its representation abilities and the properties of algebraic operations
on it have been widely studied
\cite{Berth:2010:CAN:1941063,DBLP:journals/corr/cs-OH-9903005,Allouche:2003:AST:861371}. An
{\it abstract numeration system} has been proposed by Lecomte and Rigo
in 1999 \cite{DBLP:journals/corr/cs-OH-9903005}.

\begin{definition}[\cite{Berth:2010:CAN:1941063}]
Assume that $(A, <)$ is a totally ordered set of symbols. Then, the set $A^*$ is
 totally ordered by {\it radix order} $\prec$
\footnote{Sometimes called {\it length-lexicographic order} or {\it
 genealogical order}.} defined as follows. Let $u, v$ be two strings in
 $A^*$. We write $u \prec v$ if either $|u| < |v|$ or $|u| = |v|$,
 and there exists $r,s,t \in A^*, \; a, b \in A$ with $u = ras, v =
 rbt$ and $a < b$. By $u \preceq v$, we mean that either $u \prec v$ or
 $u = v$. The set $A^*$ can be totally ordered by lexicographic
 order. \ed
\end{definition}
\begin{definition}[\cite{Berth:2010:CAN:1941063}]\label{def:ans}
An ANS $\CS$ is a triple $(L, A, <)$ comprising a totally ordered finite set of symbols
 $(A, <)$ and an infinite regular language $L$ over $(A, <)$. The 
 order isomorphism  map $\rep_\CS: \nat \rightarrow L$ is the one-to-one
 correspondence mapping $n \in \nat$ to the $(n+1)$-th string
 in the radix-ordered language $L$. The inverse map is denoted by
 $\val_\CS: L \rightarrow \nat$. \ed
\end{definition}

An ANS is not only simple but also flexible.
The following example shows that an ANS can represent ordinary
$2$-ary (binary) or $10$-ary (decimal) systems.

\begin{example}
Consider the ANS $\CS_2 = (L, \{0,1\}, <)$ such that $L = \{0\} \cup
 \{1\}\{0,1\}^*,$\\
 $0 < 1$. The strings in $L$ will be ordered by radix order as $\{0, 1, 10, 11, 100, \ldots\}$, with $\CS_2$ generating the binary system in this manner. 
That is, the mapping:
\begin{eqnarray*}
&&\rep_{\CS_2}(0) = 0, \; \rep_{\CS_2}(1) = 1, \; \rep_{\CS_2}(2) = 10,
 \rep_{\CS_2}(3) = 11, \; \rep_{\CS_2}(4) = 100, \; \ldots\\
&&\val_{\CS_2}(0) = 0, \;\; \val_{\CS_2}(1) = 1, \; \val_{\CS_2}(10) = 2,
 \; \val_{\CS_2}(11) = 3, \; \val_{\CS_2}(100) = 4, \; \ldots
\end{eqnarray*}
makes a correspondence between $\nat$ and its binary representation.
\ed
\end{example}

In fact, in terms of {\it recognizability}, it is known that an ANS is
more powerful than classical $k$-ary systems.
More details about ANSs are available from the work of Lecomte and 
Rigo \cite{Berth:2010:CAN:1941063}.
\section{Algorithms for ANSs using matrices and vectors}\label{sec:algorithm}
In this section, we provide algorithms for $\rep_\CS$ and $\val_\CS$.
The main mechanisms in these algorithms are sum and product operations on
a trim DFA and its matrix representation.
The concepts behind the algorithms in this section are described in
the work of Lecomte and Rigo \cite{Berth:2010:CAN:1941063}.
Here, we interpret these algorithms in terms of a matrix representation and analyze
their complexity.

\begin{figure}[t]
\begin{minipage}[t]{0.5\columnwidth}
\setlength{\doublerulesep}{.4pt}
\hrule\hrule\hrule\vspace{0.1cm}
\vspace{0.04cm}{\bf Algorithm 1 $\val_\CS$}
\vspace{0.1cm}\hrule
\input{img/val}
\vspace{1.57cm}\hrule
\end{minipage}
\begin{minipage}[t]{0.5\columnwidth}
\setlength{\doublerulesep}{.4pt}
\hrule\hrule\hrule\vspace{0.1cm}
{\bf Algorithm 2 $\rep_\CS$ (\cite{Berth:2010:CAN:1941063}, modified)}
\vspace{0.1cm}\hrule
\input{img/rep}
\vspace{0.1cm}\hrule
\end{minipage}
\end{figure}

\begin{table*}[t]
\centering{
\caption{The complexities of Algorithm 1 and Algorithm 2}
\label{table:complexity}
\setlength{\doublerulesep}{.4pt}
\begin{tabular}[t]{c | c | c c c}
\hline\hline\hline
 & Time complexity & M--M & M--V & V--V
 \\\hline
$\val_\CS(w)$ \; & $\Theta(|w| \times \mult(|w|) \times |\CD|^2) $ & $0$ &
			 \; $|w| - 1$ \; & 1\\
$\rep_\CS(n)$ \; & \; $\order(\ell \log \ell \times \mult(\ell \log \ell) \times \mmult(|\CD|))$ \; & \; $\order(\ell \log \ell)$ \; & $0$ & $\order(\ell \times |A|)$\\\hline
\multicolumn{5}{l}{Let $\ell = |\rep_\CS(n)|$. M--M, M--V, and V--V denote
 the number of matrix--matrix, }\\
 \multicolumn{5}{l}{matrix--vector, and vector--vector multiplications, respectively.}
\end{tabular}
}
\end{table*}

\label{sec:complexity}
We omit the details of Algorithm 1 ($\val_\CS$) and
Algorithm 2 ($\rep_\CS$) for reasons of brevity.
Table \ref{table:complexity} shows the complexities of Algorithm 1 and
Algorithm 2. 
The symbol $\mult(n)$ and $\mmult(n)$ denote the complexity of $n$-bit integer
multiplication and the number of factor multiplications in the $n \times
n$ matrix multiplication.
The implementation of these algorithms has been published
by the author
\footnote{\href{http://sinya8282.github.com/RANS/}{http://sinya8282.github.com/RANS/}}.

\begin{remark}\label{rem:multiplication}
In practice, there are efficient integer-multiplication algorithms such as the
 Sch\"onhage--Strassen algorithm ($\mult(n) = \order(n \log n \log\log
 n)$) or the F\"{u}rer algorithm ($\mult(n) = 2^{\order(\log^*{n})}$) \cite{Furer:2007:FIM:1250790.1250800}.
The matrix multiplication in the algorithm for $\rep_\CS$ is the
 bottleneck. The Strassen algorithm ($\mmult(n) =
 \order(n^{2.807})$) is widely known and can be used as an efficient matrix-multiplication
 algorithm \cite{Golub:1996:MC:248979}.
\ed
\end{remark}
\section{ANSs and compression}\label{sec:compression}
In this section, we extend the notion of a one-to-one correspondence
between $\nat$ and a language to one between two regular languages, which involves a base conversion in ANSs. Moreover, we investigate
its application to compression, especially its compression ratio.

\subsection{Base conversion in ANSs}\label{sec:conversion}
A base conversion $\conv{\CS}{\CS'}: L \rightarrow L'$ is an order
isomorphism function between L and L'.

\begin{definition}\label{def:conv}
Consider two ANSs, $\CS = (L, A, <)$ and $\CS' = (L', A', <')$. A base conversion $\conv{\CS}{\CS'}$ between $\CS$ and $\CS'$ is defined as:
\[
\Conv{\CS}{\CS'}(w) := \rep_{\CS'}(\val_\CS(w)).
\]
\ued
\end{definition}

\begin{remark}
Consider two ANSs, $\CS=(L, A, <)$ and $\CS'= (L', A', <')$, and two DFAs,
 $\CD$ and $\CD'$, such that $L = L(\CD), L' = L(\CD')$, and a string $w \in L$.
For the number $\ell = |\rep_{\CS'}(\val_\CS(w))|$, the complexity of a
 base conversion $\conv{\CS}{\CS'}$ is as follows (see Table \ref{table:complexity}):
\begin{eqnarray*}
\Theta(|w| \times \mult(|w|) \times |\CD|^2) + \order(\ell \log \ell
 \times \mult(\ell \log \ell) \times \mmult(|\CD'|)).
\end{eqnarray*}
\ued
\end{remark}

\subsection{Compression ratio}\label{sec:compression_ratio}
In general, the compression ratio is defined as the ratio of the length of the input string to
that of the output (compressed) string. This notion can be applied
naturally to a base conversion in ANSs.

\begin{definition}\label{def:ratio}
Consider two ANSs, $\CS = (L, A, <)$ and $\CS' = (L', A', <')$. The compression ratio
$\ratio$ of a base conversion $\conv{\CS}{\CS'}$ is defined as follows.
\begin{itemize}
\item For a string $w \in L: \;\;\;\; \ratio(\conv{\CS}{\CS'}, w) := \left|\conv{\CS}{\CS'}(w)\right| / \left|w\right|.$
 \item For a number $n \in \nat: \;\; \ratio(\conv{\CS}{\CS'}, n) := \ratio(\conv{\CS}{\CS'}, \max_{\prec}(L \cap A^n)).$\\
Note that $\max_{\prec}(L \cap A^n$) is the maximum, under the radix
	   order $\prec$, string that has the length $n$. 
\item For a regular language $L$, the limit of compression ratio is
	  defined as:
\[
 \ratio(\conv{\CS}{\CS'}) := \displaystyle\lim_{|w \in L|\rightarrow
	  \infty} \ratio(\conv{\CS}{\CS'}, w).
\]
We note that if the limit of compression ratio 
	  converges, then its limit value coincides with the average
	  compression ratio of $\conv{\CS}{\CS'}$:
\[
\displaystyle\lim_{|w \in L|\rightarrow \infty}
	  \ratio(\conv{\CS}{\CS'}, w) = \displaystyle\lim_{n \rightarrow
	  \infty}\frac{1}{n}\sum_{i=0}^{n}\ratio(\conv{\CS}{\CS'}, \rep_\CS(i)).
\]
\end{itemize}
The base conversion $\conv{\CS}{\CS'}$ is a ``compression'' if and only
 if $\ratio(\conv{\CS}{\CS'}) \leq 1$.\ed
\end{definition}

At the present time, it is not clear whether the limit of compression ratio of a base conversion
$\ratio(\conv{\CS}{\CS'})$ converges or diverges.

We now introduce Theorem \ref{theorem:ratio} to clarify the limit
value. This requires following Lemma \ref{lemma:counting_} about counting
functions.

\begin{lemma}
\label{lemma:counting_}
For any regular language $L$, the following equation holds:
\begin{eqnarray*}
C_L^{\leq}(n) =
\begin{cases}
\Theta(C_L(n)) = \Theta(n^{\pd{L}} \ind(L)^n) & \mathrm{if} \;
 \ind(L) > 1, \;\;\;\;\;\; \mathrm{(i)} \\
\Theta(n \times C_L(n)) = \Theta(n^{\pd{L}+1}) & \mathrm{if} \;
 \ind(L) = 1. \;\;\;\;\;\; \mathrm{(ii)}
\end{cases}
\end{eqnarray*}
\end{lemma}

\begin{proof}
In case (i), the equation $C_L^{\leq}(n) = \bar\Omega(n^\pd{L} \ind(L)^n)$ always
 holds because $C_L(n) \leq C_L^{\leq}(n)$ by definition. From the
 monotonicity of the function $C_L^{\leq}$, the following equations hold:
\begin{eqnarray}
C_L^{\leq}(n) &=& \Omega(n^\pd{L} \ind(L)^n),\label{lem2:eq1}\\
C_L^{\leq}(n) &=& \order\left( \sum_{i=0}^n i^\pd{L} \ind(L)^i
						\right)
 = \order\left(n^\pd{L} \sum_{i=0}^n \ind(L)^i\right)\nonumber\\
 &=& \order(n^\pd{L} \ind(L)^n)\label{lem2:eq2}.
\end{eqnarray}
From Equations \eqref{lem2:eq1} and \eqref{lem2:eq2}, the upper and
 lower bounds are equal.
\[
\so \;\;\;\; C_L^{\leq}(n) = \Theta(n^\pd{L} \ind(L)^n).
\]

In case (ii), by Theorem \ref{theorem:shur}, we can obtain the following equation:
\begin{eqnarray}
C_L^{\leq}(n) = \bar\Theta\left( \sum_{i=0}^n i^\pd{L} \ind(L)^i
							\right)
 = \bar\Theta\left(\sum_{i=0}^n i^\pd{L}\right)
 = \bar\Theta(n^{\pd{L}+1}). \label{lem2:eq3}
\end{eqnarray}
Note that the last transformation in Equation \eqref{lem2:eq3} uses a power sum formula. 
Clearly, $C_L^{\leq}(n) = \Omega(n^{\pd{L}+1})$ holds from the
 monotonicity of the function $C_L^{\leq}$.
\[
 \so \;\;\;\; C_L^{\leq}(n) = \Theta(n^{\pd{L}+1}).
\]
\ued
\end{proof}

\begin{theorem}[The compression ratio theorem]
\label{theorem:ratio}
Consider two ANSs, $\CS = (L, A, <)$ and $\CS' = (L', A', <')$. The average
 compression ratio of a base conversion $\conv{\CS}{\CS'}: L
 \rightarrow L'$ satisfies the following conditions:
 \begin{enumerate}
   \renewcommand{\labelenumi}{(\roman{enumi})}
  \item If $\ind(L') > 1$, then $\ratio(\conv{\CS}{\CS'}) = \log\ind(L)
		 / \log\ind(L')$,
   \item If $\ind(L) > \ind(L') = 1$, then $\ratio(\conv{\CS}{\CS'}) = \infty$,
   \item If $\ind(L) = \ind(L') = 1$, then:
		 \begin{eqnarray*}
		   \ratio(\conv{\CS}{\CS'}) =
		   \begin{cases}
			 0 & \text{if} \;\;\; \pd{L} < \pd{L'},\\
			 \text{in a finite interval} & \text{if} \;\;\; \pd{L} = \pd{L'},\\
			 \infty & \text{if} \;\;\; \pd{L} > \pd{L'}.
		   \end{cases}
		 \end{eqnarray*}
 \end{enumerate}
\end{theorem}

\begin{proof}
Consider $w \in L$ and $w' \in L'$ such that $w' = \conv{\CS}{\CS'}(w)$. That
 is, $\val_\CS(w) = \val_\CS'(w)$. From Definition \ref{def:ratio}, the following
 equation holds:
\begin{eqnarray*}
\ratio(\conv{\CS}{\CS'}) = \displaystyle\lim_{|w \in L| \rightarrow \infty} \frac{\left|w'\right|}{\left|w\right|}.
\end{eqnarray*}
Moreover, in the case of $\lim_{|w \in L| \rightarrow
 \infty}$, for each value of $\val_{\CS}(w),\val_{\CS'}(w')$ and $|w'|$ diverge
 to infinity. This is because the following inequalities:
\begin{eqnarray}
C_L^{\leq}(|w|-1) \leq \val_{\CS}(w) < C_L^{\leq}(|w|)\label{thm2:eq1},\\
C_L'^{\leq}(|w'|-1) \leq \val_{\CS'}(w') < C_L'^{\leq}(|w'|)\label{thm2:eq2}
\end{eqnarray}
clearly always hold, from Definition \ref{def:ans}.
We can then give a proof for each of the conditions (i)--(iii) based on
 Inequalities \eqref{thm2:eq1} and \eqref{thm2:eq2} and the following Equation \eqref{thm2:eq3}:
\begin{eqnarray}
\val_{\CS}(w) = \val_{\CS'}(w')\label{thm2:eq3}.
\end{eqnarray}
In the following derivations, we assume that $|w|$ is sufficiently large for
 Theorem 1 to apply.

Considering condition (i), the following inequality is obtained by applying
 Lemma \ref{lemma:counting_} in Equation \eqref{thm2:eq2} and using some
 constants $u', l'$:
\begin{eqnarray*}
l'|w'|^{\pd{L'}}\ind(L')^{|w'|} \leq \val_{\CS'}(w') \leq u'|w'|^{\pd{L'}}\ind(L')^{|w'|}.
\end{eqnarray*}
Similarly, assuming that $\ind(L) > 1$, the following inequality is obtained by applying Equation \eqref{thm2:eq1} and using some constants $u, l$:
\begin{eqnarray*}
l|w|^{\pd{L}}\ind(L)^{|w|} \leq \val_{\CS}(w) \leq u|w|^{\pd{L}}\ind(L)^{|w|}.
\end{eqnarray*}

There will then exist two positive real numbers $x,x'
 \in \real$ such that $l \leq x \leq u, l' \leq x' \leq u'$, and\[
x|w|^{\pd{L}}\ind(L)^{|w|} = x'|w'|^{\pd{L'}}\ind(L')^{|w'|} 
\]
because of Equation \eqref{thm2:eq3}. The following equation is obtained by
 taking the logarithm of the above equation:
\[
  |w|\log{\ind(L)} + \pd{L}\log{|w|} + \log{x}
= |w'|\log{\ind(L')} + \pd{L'}\log{|w'|} + \log{x'}.
\]
Dividing this equation by $|w'|\log\ind(L') \neq 0$, the following
 equation is obtained:
\[
\frac{|w|\log{\ind(L)}}{|w'|\log{\ind(L')}} +
 \frac{\pd{L}\log{|w|} + \log{x}}{|w'|\log{\ind(L')}} 
= 1 + \frac{\pd{L'}\log{|w'|} + \log{x'}}{|w'|\log{\ind(L')}}.
\]
The second terms of both sides of the above equations converge to zero in
 the limit $\lim_{|w| \rightarrow \infty}$. The following equation therefore holds:
\[
\frac{|w|}{|w'|}\frac{\log{\ind(L)}}{\log{\ind(L')}} = 1.
\]
The factor $|w|/|w'|$ on the left-hand side is the reciprocal of
 the average compression ratio. 
\[
 \so \;\;\;\;\;\; \ratio(\conv{\CS}{\CS'}) = \frac{\log\ind(L)}{\log\ind(L')}.
\]
Similarly, it is easy to prove that $\ratio(\conv{\CS}{\CS'})$ converges
 to zero in the case of $\ind(L) = 1$.
Considering condition (ii), we can also
 prove that $\ratio(\conv{\CS}{\CS'})$ diverges to infinity, by a similar
 argument.

Considering condition (iii), by a similar argument to that for condition (i),
there will exist two positive real numbers $x, x'
 \in \real$ such that
 \[
 x|w|^{\pd{L}+1} = x'|w'|^{\pd{L'}+1}.					   
 \]
The following equation is obtained by
 dividing this equation by $x'|w|^{\pd{L'}+1}$:
\[
\frac{x}{x'} |w|^{\pd{L}-\pd{L'}} = \frac{|w'|^{\pd{L'}+1}}{|w|^{\pd{L'}+1}}
 = \left( \frac{|w'|}{|w|} \right)^{\pd{L'}+1}.
\]
The right-hand side is the $\pd{L'}+1$-th power of the average
 compression  ratio. Therefore, the following holds:
\begin{eqnarray}
\sqrt[\pd{L'}+1]{\frac{x}{x'} |w|^{\pd{L}-\pd{L'}}} = \ratio(\conv{\CS}{\CS'}).\label{thm2:eq4}
\end{eqnarray}
In the limit $\lim_{|w| \rightarrow\infty}$, the left-hand side
 of Equation \eqref{thm2:eq4} is:
\[
   \ratio(\conv{\CS}{\CS'}) =
	\begin{cases}
	0 & \text{if} \;\; \pd{L} < \pd{L'},\\
	\sqrt[\pd{L'}+1]{x/x'} & \text{if} \;\; \pd{L} = \pd{L'},\\
	\infty & \text{if} \; \pd{L} > \pd{L'}.
	\end{cases}
\]
Although the positive real numbers $x$ and $x'$ are in a finite
 interval, we cannot make further assumptions about convergence in the
 case of $\pd{L} = \pd{L'}$.
\qed
\end{proof}
\section{Block-based compression}\label{sec:application}\label{sec:sepcomp}
For a language with exponential growth, large integer arithmetic  is
required, particularly for the computations $\val_\CS: L
\rightarrow \nat$ and its inverse $\rep_\CS$. In a real
machine, large integer multiplication implies a correspondingly large cost 
(see Remark \ref{rem:multiplication}).
This implies that ANS-based compression for a large input string
may be impractical. The goal of this section is to propose a block-based
compression method for addressing this drawback via a {\it factorial language}.

\begin{definition}\cite{Sakarovitch:2009:EAT:1629683,Shur:2008:CCR:1813695.1813728}
We denote the set of all substrings in $L$ as $\fac(L)$.
A regular language $L$ is called a {\it factorial language} if and only if $L =
 \fac(L)$. \ed
\end{definition}

\begin{remark}\label{remark:factorial}
For any regular language $L$, its factorial language $\fac(L)$
 is also regular. Consider an automaton $\CA = (Q, A, \delta, I, F)$. The
 automaton $\CA' = (Q, A, \delta, Q, Q)$ will recognize the
 factorial language $\fac(L(\CA))$. \ed
\end{remark}

\begin{fact}[\cite{Shur:2008:CCR:1813695.1813728}]
\label{fact:factorial}
If a regular language $L$ has the complexity $\bar\Theta(f(n))$ for
some function $f$, then the language $\fac(L)$ has the complexity
$\bar\Theta(f(n))$ as well. \ed
\end{fact}

Any substring of a string in factorial language $L$ is also in $L$.
We can therefore partition the input string freely if the base conversion has the
domain of a factorial language. 
Algorithm 3 shows the application of block-based compression using a
factorial language. (Here, we assume that the length of the input string is a
multiple of the block length. That is, $|w| = m\ell$ for some $m$.)

\begin{figure}[t]
\setlength{\doublerulesep}{.4pt}
\hrule\hrule\hrule\vspace{0.1cm}
{\bf Algorithm 3: Block-based compression using a factorial language}
\vspace{0.1cm}\hrule\vspace{0.1cm}
\input{img/sep}
\vspace{0.1cm}\hrule
\end{figure}
\begin{figure}[t]
\begin{eqnarray*}
&&w\; = \overbrace{w_1 w_2 \ldots w_{\ell}}^{w_{[1]}}
 \overbrace{w_{\ell + 1} w_{\ell + 2} \ldots w_{2\ell}}^{w_{[2]}}
 w_{2\ell+1} \ldots
 w_{m\ell - \ell} \overbrace{w_{m\ell - \ell + 1}
 w_{m\ell - \ell + 2} \ldots w_{m\ell}}^{w_{[m]}}\\
&&
 \;\;\;\;\;\;\;\;\;\;\;\;\;\;\;\;\;\;\;\;\;\;\;\;\;\;\;\;\;\;\;\;\;
 \;\;\;\;\;\;\;\;\;\;\;\;\;\;\;\;\;\;\;\;\;\;\;\;\;
 \Downarrow\\
&& w' = \ell_{1} \; \conv{\fac(\CS)}{\CS'}(w_{[1]}) \;
 \ell_{2} \; \conv{\fac(\CS)}{\CS'}(w_{[2]}) \; \ldots \; \ell_{m}
 \; \conv{\fac(\CS)}{\CS'}(w_{[m]})
\end{eqnarray*}\vspace{-.7cm}
\caption{An image of the input and output for a block-based compression}
\label{fig:block}
\end{figure}

\begin{remark}\label{remark:sepcomp}
In Step 4 of Algorithm 3, each length of $w_{[i]}$ is set by $\ell$.
Each length of $\conv{\fac(\CS)}{\CS'}(w_{[i]})$ is therefore in the
 interval $[\ell_{min}, \ell_{max}]$, where
 \[
  \ell_{max} = \max_{w \in A^\ell} |\conv{\fac(\CS)}{\CS'}(w)|, \;\;\;\;
 \ell_{min} = \min_{w \in A^\ell} |\conv{\fac(\CS)}{\CS'}(w)|.
 \]
 Therefore, each description length of $\ell_{i}$ will be $\log(\ell_{max}-\ell_{min}+1)$.
\ed
\end{remark}

\begin{theorem}\label{theorem:sepcomp}
Consider a base conversion $\conv{\CS}{\CS'}: L \rightarrow L'$. The
 corresponding block-based compression using $\fac(\CS)$ will have the
 same compression ratio if the block length $\ell$ is sufficiently
 large.
\end{theorem}

\begin{proof}
For any regular languages $L$ and $L'$, $\ratio(\conv{\CS}{\CS'}) =
 \ratio(\conv{\fac(\CS)}{\CS'})$ will hold, following Fact \ref{fact:factorial}.
 In addition, each description length of $\ell_{i}$ is 
$\log(\ell_{max}-\ell_{min}+1)$, following Remark \ref{remark:sepcomp}. It will be
 constant if $\ell$ is fixed.
 The description length can therefore be ignored whenever $\ell$ is sufficiently large.
\qed
\end{proof}
\section{Conclusion}\label{sec:conclusion}
We proposed and investigated ANS-based compression in
\S \ref{sec:compression} and \S \ref{sec:application}.
The main idea of an ANS-based compression is to use a base conversion in the ANS.
We analyzed the complexity of the primitives of ANSs, namely $\rep_\CS$ and
$\val_\CS$ (see Table \ref{table:complexity}).
In addition, we clarified the behavior of the average compression ratio for
ANS-based compression in Theorem \ref{theorem:ratio}.
Finally, we proposed a block-based compression method that does not reduce the
compression ratio in \S \ref{sec:application}.

The author has implemented an ANS-based compression library
as open-source software 
\footnote
{\href{http://sinya8282.github.com/RANS/}{http://sinya8282.github.com/RANS/}}.
This will enable others to verify the present paper's results.\\

{\bf Acknowledgement} My deepest appreciation goes to my advisor, Prof. Sassa (Tokyo Institute of Technology), whose enormous support and insightful
comments were invaluable during the course of my study. 
Special thanks also go to Prof. Rigo (Universit\'e de Li\`ege), whose
comments and publications \cite{Berth:2010:CAN:1941063} have helped me
very much throughout the production of this study.

\bibliographystyle{splncs.bst}
\bibliography{ref}
\end{document}